\documentstyle[extras,11pt]{article}

\newcommand{\xx}{\mbox{\boldmath $x$}}

\newcommand{\vv}{\mbox{\boldmath $v$}}
\newcommand{\pp}{\mbox{\boldmath $p$}}
\newcommand{\bb}{\mbox{\boldmath $b$}}
\newcommand{\cc}{\mbox{\boldmath $c$}}
\newcommand{\dd}{\mbox{\boldmath $d$}}
\newcommand{\ee}{\mbox{\boldmath $e$}}

\newcommand{\AAA}{\mbox{\boldmath $A$}}

\newcommand{\ga}{\mbox{$\gamma$}}
\newcommand{\al}{\mbox{$\alpha$}}
\newcommand{\bt}{\mbox{$\beta$}}

\newcommand{\In}{\mbox{\rm in}}
\newcommand{\Out}{\mbox{\rm out}}

\newcommand{\la}{\leftarrow}
\newcommand{\Q}{\mbox{\rm\bf Q}}
\newcommand{\Qplus}{\Q^+}
\newcommand{\minimize}{\mbox{\rm minimize }}
\newcommand{\maximize}{\mbox{\rm maximize }}

\newcommand{\st}{\mbox{\rm subject to }}
\newcommand{\CN}{\mbox{${\cal N}$}}
\newcommand{\CS}{\mbox{${\cal S}$}}
\newcommand{\CM}{\mbox{${\cal M}$}}
\newcommand{\CG}{\mbox{${\cal G}$}}

\newcommand{\R}{\mbox{\rm\bf R}}
\newcommand{\Rplus}{\R_+}

\newcommand{\lefta}{\mbox{${\leftarrow}$}}

\makeatletter
\def\fnum@figure{{\bf Figure \thefigure}}
\def\fnum@table{{\bf Table \thetable}}

\long\def\@mycaption#1[#2]#3{\addcontentsline{\csname
 ext@#1\endcsname}{#1}{\protect\numberline{\csname
  the#1\endcsname}{\ignorespaces #2}}\par
     \begingroup
       \@parboxrestore
          \small
       \@makecaption{\csname fnum@#1\endcsname}{\ignorespaces
#3\endgroup}
      }

\newcommand{\ADNB}{{\bf ADNB}}
\def\G2{{\bf DG2}}

\begin{document}

\title{2-Player Nash and Nonsymmetric Bargaining Games:\\
Algorithms and Structural Properties}

\author{
{\Large\em Vijay V. Vazirani}\thanks{College of Computing,
Georgia Institute of Technology, Atlanta, GA 30332--0280,
E-mail: {\sf vazirani@cc.gatech.edu}.
}
}

\date{}
\maketitle

\begin{abstract}
The solution to a Nash or a nonsymmetric bargaining game is obtained by maximizing a concave function over 
a convex set, i.e., it is the solution to a convex program. We show that each 2-player game whose convex
program has linear constraints, admits a rational solution and such a solution can be found in polynomial time
using only an LP solver. If in addition, the game is succinct, i.e., the coefficients in its convex program are
``small'', then its solution can be found in strongly polynomial time. We also give a non-succinct linear game whose solution
can be found in strongly polynomial time.

\end{abstract}

\section{Introduction}
\label{sec.intro}

In game theory, 2-player games occupy a special place -- not only because numerous applications
involve 2 players but also because they often have remarkable properties that are not possessed by 
extensions to more players.

For instance, in the case of Nash equilibrium, the 2-player case is the most extensively studied and used,
and captures a rich set of possibilities, e.g., those encapsulated in canonical games such as
prisoner's dilemma, battle of the sexes, chicken, and matching pennies. In terms of properties, 2-player Nash 
equilibrium games always have rational solutions whereas games with 3 or more players may have only irrational solutions;
an example of the latter, called ``a three-man poker game,'' was given by Nash \cite{nash.equilibrium}).
Finally, von Neumann's minimax theorem for 2-player zero-sum games yields a polynomial time algorithm using LP.
On the other hand, 3-player zero-sum games are PPAD-hard, since 2-player non-zero-sum games can be reduced to them; the
reduction is due to \cite{von} and PPAD-hardness is due to \cite{CDT}. 

John Nash's seminal paper defining the bargaining game dealt only with the case of 2-players \cite{nash.bargain}. 
Later, it was observed that his entire setup, and theorem characterizing the bargaining solution, easily
generalize to the case of more than 2 players, e.g., see \cite{Kalai.nonsymmetric}.

Recently, Vazirani \cite{va.NB} initiated a systematic algorithmic study of Nash bargaining games and also
carried this program over to solving nonsymmetric bargaining games of Kalai \cite{Kalai.nonsymmetric}.
In this paper we carry the program further, though only for the case of 2-player games.
Our findings indicate that this case exhibits a rich set of possibilities algorithmically and calls for a
further investigation.

The solution to a Nash or a nonsymmetric bargaining game is obtained by maximizing a concave function over 
a convex set, i.e., it is the solution to a convex program. Two basic classes of these games defined in \cite{va.NB}
are NB and LNB.
The convex program for a game in NB admits a polynomial time separation oracle and hence its
solution can be obtained to any desired accuracy using the ellipsoid algorithm.
All constraints in the convex program for a game in LNB are linear and \cite{va.NB} gives combinatorial
polynomial time algorithms for several games in this class; by a {\em combinatorial algorithm} we mean an
algorithm that performs an efficient search over a discrete space. 
Let NB2 and LNB2, respectively, be the restrictions of these classes to 2-players games.

We show that for solving any game in LNB2, it is not essential to solve a convex program --  an LP solver suffices.
As a consequence, all games in LNB2 have rational solutions; this property does not hold for 3-player games in LNB. 
We then define a subclass of LNB2 called SLNB2, consisting of {\em succinct} games, i.e., the coefficients
in its convex program are ``small''. We show that all games in SLNB2 admit strongly polynomial algorithms; however, these
algorithms are not combinatorial.
This class includes nontrivial and interesting games, e.g., the game \G2, which consists of a directed graph
with edge capacities and each player is a source-sink pair desiring flow (see Section \ref{sec.classes} for definition).
This game is derived from Kelly's {\em flow markets} \cite{Kelly}.

Next, we ask if there is a game in (LNB2 - SLNB2) that admits a strongly polynomial time algorithm.
The answer turns out to be ``yes''.  We show that the 2-player version of the game \ADNB, for which a 
combinatorial polynomial time algorithm is given in \cite{va.NB}, admits a combinatorial strongly polynomial algorithm.
This game is derived from the linear case of the Arrow-Debreu market model (see Section \ref{sec.ADNB2} for definition).

Finally, we ask if there is a game in (NB2 - LNB2) that can be solved in polynomial time without a convex
program solver. Once again, the answer turns out to be ``yes''.  We give a game whose solution reduces to solving
a degree 4 equation. Alternatively, it also admits an elegant geometric solution.

Our last 2 results raise interesting questions, e.g., is there a characterization of the subclass of LNB2 which consists of all
games that admit strongly polynomial algorithms? They also indicate that the class NB2, in particular (NB2 - LNB2), 
may be worth exploring further algorithmically and structurally, e.g., does (NB2 - LNB2) contain a game that always has a
rational solution? And are there 2-player games, not in NB2, whose solution can be computed in polynomial time?

Building on a remarkable convex program of Eisenberg and Gale \cite{eisenberg}, \cite{JV.EG} gave the notion of {\em Eisenberg-Gale markets};
see Section \ref{sec.model}. In answering an open question of \cite{JV.EG} affirmatively, \cite{CDV.EG} showed that EG(2) markets,
i.e., the restriction of Eisenberg-Gale markets to 2 buyers, always admit a rational solution and it can be found using only an LP solver.
Our first result is obtained by extending their algorithm. 
For the second result, we use the notion of a {\em flexible budget market} given in \cite{va.NB}; see Section \ref{sec.ADNB2}.
We reduce \ADNB2 to such a market and give a combinatorial algorithm for finding an equilibrium in it.

\section{Nash and Nonsymmetric Bargaining Games}
\label{sec.Nash}

An {\em $n$-person Nash bargaining game} consists of a pair $(\CN, \cc)$, where $\CN \subseteq \Rplus^n$
is a compact, convex set and $\cc \in \CN$. Set $\CN$ is the {\em feasible set} and its elements give
utilities that the $n$ players can simultaneously accrue. Point $\cc$ is the {\em disagreement point}
-- it gives the utilities that the $n$ players obtain if they decide not to cooperate.
The set of $n$ agent will be denoted by $B$ and the agents will be numbered $1, 2, \ldots n$.
Game $(\CN, \cc)$ is said to be {\em feasible} if there is a point
$\vv \in \CN$ such that $\forall i \in B, \ v_i > c_i$.

The solution to a feasible game is the point $\vv \in \CN$ that satisfies
the following four axioms:
\begin{enumerate}
\item
{\bf Pareto optimality:}  No point in $\CN$ can weakly dominate $\vv$.
\item
{\bf Invariance under affine transformations of utilities:} 
\item
{\bf Symmetry:} The numbering of the players should not affect the solution.
\item
{\bf Independence of irrelevant alternatives:} If $\vv$ is the solution for $(\CN, \cc)$, and
$\CS \subseteq \Rplus^n$ is a compact, convex set satisfying $\cc \in \CS$ and  $\vv \in \CS \subseteq \CN$,
then $\vv$ is also the solution for $(\CS, \cc)$.
\end{enumerate}

Via an elegant proof, Nash proved:

\begin{theorem}
{\bf Nash} \cite{nash.bargain}
\label{thm.nash}
If game $(\CN, \cc)$ is feasible
then there is a unique point in $\CN$ satisfying the axioms stated above.
This is also the unique point that maximizes $\Pi_{i \in B}  {(v_i - c_i)}$, over all $\vv \in \CN$.
\end{theorem}

Thus Nash's solution involves maximizing a concave function over
a convex domain, and is therefore the optimal solution to the convex program that maximizes
$\sum_{i \in B}  \log (v_i - c_i)$ subject to $\vv \in \CN$.
As a consequence, if for a specific game, a separation oracle can be implemented
in polynomial time, then using the ellipsoid algorithm one can get as good an approximation to
the solution as desired \cite{GLS}. 

Kalai \cite{Kalai.nonsymmetric} generalized Nash's bargaining game by removing the axiom
of symmetry and showed that any solution to the resulting
game is the unique point that maximizes $\Pi_{i \in B}  {(v_i - c_i)^{p_i}}$, over all $\vv \in \CN$,
for some choice of positive numbers $p_i$, for $i \in B$, such that $\sum_{i \in B} {p_i} \ = \ 1$.
Thus, any particular nonsymmetric bargaining solution is specified by giving the
$p_i$'s satisfying the 2 conditions given above.
For the purposes of computability, we will restrict to rational $p_i$'s.
Equivalently, let us define the {\em n-person nonsymmetric bargaining game} as follows. Assume that
$B, \CN, \cc$ are as defined above. In addition, we are given the {\em clout\footnote{The
choice of the term ``clout of a player'' is justified by a theorem of Kalai
stating that the solution to this game corresponds precisely to the solution of
a $k$-person game, with $k = \sum_{i \in B} {w_i}$, which is obtained by taking $w_i$ copies
of player $i$, for $1 \leq i \leq n$.} } of
each player: a positive integer $w_i$ for each player $i$. Assuming the game is feasible, the
solution to this nonsymmetric bargaining game is
the unique point that maximizes $\Pi_{i \in B}  {(v_i - c_i)^{w_i}}$, over all $\vv \in \CN$.

One more remark is in order. As shown by Kalai \cite{Kalai.nonsymmetric}, any nonsymmetric game can be reduced to a Nash bargaining game 
over a larger number of players. However, this reduction is not useful for our purpose because once the number of players
increases, the special properties of 2-player games are lost.

\section{The Classes NB2, LNB2 and SLNB2}
\label{sec.classes}

Before defining the classes NB2 and LNB2, we recall the definition of the classes NB and LNB from \cite{va.NB}.
Let $\CG$ be an $n$-person Nash or nonsymmetric bargaining game whose solution is given by the optimal solution to the
following convex program, where $\xx$ are $m$ auxiliary variables, the functions $f_i$ are convex.
(Clearly, $\CG$ is a Nash bargaining game if each $w_i = 1$.)

\begin{lp}
\label{CP-NB}
\maximize & \sum_{i \in B} w_i  \log (v_i - c_i)  \\[\lpskip]
\st       & \mbox{for} \ \ i = 1 \ldots k: \ \ \ \   f_i(\vv, \xx) \leq   0   \nonumber  \\
          & \ \ \ \ \ \ \ \ \ \ \ \ \ \ \ \ \ \   \vv \geq 0   \nonumber \\
          & \ \ \ \ \ \ \ \ \ \ \ \ \ \ \ \ \ \   \xx \geq 0   \nonumber 
\end{lp}

The game $\CG$ is said to be in the class NB if each of the $k$ constraints of program (\ref{CP-NB})
can be checked in polynomial time at any given point $(\vv, \xx)$. This gives a separation oracle for the
program and therefore, using the ellipsoid algorithm, the Nash or nonsymmetric bargaining solution to the
game $\CG$ can be obtained to any desired accuracy, assuming the game is feasible. Furthermore, $\CG$ is
feasible iff the optimal solution to the following convex program is $> 0$, which can also be checked
in polynomial time.

\begin{lp}
\label{CP-NB-feasible}
\maximize &   \ \ \ \ \ \ \ \ \ \ \ \ \ \ \   t         \\[\lpskip]
\st       & \mbox{for} \ \ i = 1 \ldots n: \ \ \ \   v_i  \geq  c_i + t   \nonumber  \\
          & \mbox{for} \ \ i = 1 \ldots k: \ \ \ \   f_i(\vv, \xx) \leq   0   \nonumber  \\
          & \ \ \ \ \ \ \ \ \ \ \ \ \ \ \ \ \ \   \vv \geq 0   \nonumber \\
          & \ \ \ \ \ \ \ \ \ \ \ \ \ \ \ \ \ \   \xx \geq 0   \nonumber   
\end{lp}

The restriction of class NB to 2-player games yields the class NB2.

If all constraints in (\ref{CP-NB}) are linear, then game $\CG$ is said to be {\em linear}.  
If so, the constraints form a polyhedron in $\R^{n+m}$. Its projection on the first $n$ coordinates, 
corresponding to $\vv$, is a polytope, which is also the feasible set $\CN$.
The class of these games is called {\em linear Nash and nonsymmetric bargaining games}, and abbreviated to LNB.

Finally, the restriction of LNB to 2-player games gives us the class LNB2.
We will assume w.l.o.g. that the convex program for game $\CG$ in LNB2 has the following form:

\begin{lp}
\label{CP-LNB2}
\maximize & \sum_{i = 1,2} w_i  \log (v_i - c_i)  \\[\lpskip]
\st       & \ \ \ \   \AAA \xx +  \bb_1 v_1  +  \bb_2 v_2   \leq   \ee   \nonumber  \\
          & \ \ \mbox{for} \ \ i = 1, 2: \ \ \ \ \  v_i  \geq  0   \nonumber  \\
          & \ \ \ \ \ \ \ \ \ \ \ \ \ \ \ \ \ \ \ \ \ \ \ \ \ \xx \geq 0   \nonumber 
\end{lp}

where $\AAA$ is an $m \times n$ matrix, $\xx$ is a vector consisting of allocation and auxiliary variables,
say $n$ in number, and $\bb_1, \bb_2, \ee$ are $m$-dimensional vectors. 

We will say that $\CG$ is {\em succinct} if all the entries in $\AAA, \bb_1, \bb_2$ are
polynomially bounded in $m$ and $n$. The subclass of LNB2 consisting of all succinct games will 
be called SLNB2.

\section{Some Representative 2-Player Games}
\label{sec.rep}

In this section, we provide representative games for the 3 classes defined above. We will study all 3 games in detail in this paper. 

\subsection{The game DG2}
\label{sec.DG2}

The game \G2\ lies in SLNB2.
We are given a directed graph $G = (V, E)$, with $c_e \in \Qplus$ specifying the capacity of edge $e \in E$.
Two source-sink pairs are also specified, $(s_1, t_1) \ \mbox{and} \ (s_2, t_2)$. Each source-sink pair 
represents a player in the game and has its own disagreement utility (flow value) $c_i$, for $i = 1, 2$. 
In the nonsymmetric version, we are also given the clouts $w_1$ and $w_2$ of the two players.
The object is to find the Nash or nonsymmetric bargaining solution.
Let $\CG$ denote the given instance of \G2\footnote{Note that there will be no confusion in using ``$c$'' to denote 
capacities of edges as well as disagreement utilities of players since in the former case, the subscript will always be $e$ 
and in the latter case, it will be $1, 2$ or $i$.}.

Next, we give a convex program that captures the solution to $\CG$.
The flow going from $s_i$ to $t_i$ will be referred to as commodity $i$, for $i = 1, 2$, and $f_i$ will denote the total flow of commodity $i$.
For each edge $e \in E$, we have 2 variables, $f_e^1$ and $f_e^2$ which denote the amount of each commodity flowing through $e$.
The constraints ensure that the total flow going through an edge does not exceed its capacity and that for each commodity,
at each vertex, other than the source-sink pair of this commodity, flow conservation holds. 
For vertex $v \in V$, $\Out(v) = \{ (v, u) \ | \ (v, u) \in E \}$ and $\In(v) = \{ (u, v) \ | \ (u, v) \in E \}$.
The constraints of this program are simply ensuring that $(f_1, f_2)$ lies in the feasible set $\CN$.

\begin{lp}
\label{CP-2D}
\maximize &   \sum_{i = 1,2} {w_i  \log (f_i - c_i)}            \\[\lpskip]
\st       & \mbox{for} \ \ i = 1, 2: \ \ \ \   f_i = \sum_{e \in \Out(s_1)} {f_{e}^i}  \nonumber  \\
          & \forall e\in E:~   \ \   f_e^1 + f_e^2  \leq c_e  \nonumber \\
          & \mbox{for} \ \ i = 1, 2: \ \ \forall v\in V - \{s_i, t_i\}:~ \ \   \sum_{e \in \In(v)} f_e^i =  \sum_{e \in \Out(v)} f_e^i   \nonumber \\ 
          & \mbox{for} \ \ i = 1, 2: \ \ \forall e\in E:~   \ \   f_e^i  \geq 0  \nonumber 
\end{lp}

\subsection{The game ADNB2}
\label{sec.ADNB2-def}

The game \ADNB2 lies in (LNB2 - SLNB2). To define it, we first need to define
the game \ADNB, introduced in \cite{va.NB}. This game was derived from the linear case of the Arrow-Debreu model, which
differs from Fisher's linear case in that each agent comes to the market not with money but with an initial
endowment of goods. We first state it formally.

Let $B = \{1, 2, \ldots, n\}$ be a set of agents and $G = \{1, 2, \ldots, g\}$ be a set of divisible goods.
We will assume w.l.o.g. that there is a unit amount of each good. 
Let $u_{ij}$ be the utility derived by agent $i$ on receiving one unit of good $j$; w.l.o.g., we
will assume that $u_{ij}$ is integral.
If $x_{ij}$ is the amount of good $j$ that agent $i$ gets, for $1 \leq j \leq g$, then the total
utility derived by her is
\[ v_i(x) = \sum_{j \in G}  {u_{ij} x_{ij} } .\]
Finally, we assume that each agent has an initial endowment of these goods; the total amount of each
good possessed by the agents is 1 unit. 
The question is to find prices for these goods so that if each agent sells her entire initial
endowment at these prices and uses the money to buy an optimal bundle of goods, the market clears.

W.l.o.g. we may assume that each good is desired by at least one agent and each agent
desires at least one good, i.e.,
\[ \forall j \in G, \ \exists i \in B: \ u_{ij} >0 \ \ \mbox{and}  \ \
\forall i \in B, \ \exists j \in G: \ u_{ij} >0 .\]
If not, we can remove the good or the agent from consideration.

In \cite{va.NB}, we explored a different solution concept for this setting: for each agent $i$, compute the utility she accrues
from her initial endowment, say $c_i$. Let $\CN$ in $\Rplus^n$ denote the set of all possible utility vectors obtained by
distributing the goods among the agents in all possible ways. Now seek the Nash bargaining
solution for instance $(\CN, \cc)$. The setup was made more general by assuming that $c_i$'s are arbitrary numbers given with the
problem instance, i.e., they do not come from initial endowments.

Let \ADNB2 denote the restriction of this game to 2 players. We will assume these are nonsymmetric games, i.e., we are also given
the clout, $w_1$ and $w_2$ of the two players. We give a combinatorial strongly polynomial algorithm for this game; the algorithm 
in \cite{va.NB} is not strongly polynomial.

The bargaining solution to \ADNB2 is the optimal solution to the following convex program:

\begin{lp}
\label{CP-ADNB2}
\maximize & \sum_{i = 1,2} {w_i \log (v_i - c_i)}  \\[\lpskip]
\st       & \forall i = 1,2: \ \  v_i = \sum_{j \in G} {u_{ij} x_{ij} }  \nonumber  \\
          & \forall j\in G: \ \ \sum_{i = 1,2} x_{ij} \leq 1  \nonumber \\
          & \forall i = 1,2, \ \forall j\in G: \ \ x_{ij} \geq 0 \nonumber
\end{lp}

\subsection{The circle game}
\label{sec.circle}

The circle game lies in (NB2 - LNB2). Its feasible set is the intersection of the unit disk with the positive orthant.
We will consider only its Nash bargaining version. Its convex program is: 

\begin{lp}
\label{CP-circle}
\maximize & \sum_{i = 1,2} {\log (v_i - c_i)}  \\[\lpskip]
\st       & \ \ \ \ \ \           v_1^2 + v_2^2 \leq 1  \nonumber  \\
          & \forall i = 1,2: \ \   v_i \geq 0 \nonumber
\end{lp}

\section{Fisher and Eisenberg-Gale Market Models}
\label{sec.model}

We will first state Fisher's market model for the case of linear utility functions \cite{scarf}. 
Consider a market consisting of a set of $n$ buyers 
$B = \{1, 2, \ldots, n\}$, and a set of $g$ divisible goods, $G = \{1, 2, \ldots, g\}$; 
we may assume w.l.o.g. that there is a unit amount of each good. Let $m_i$ be the money 
possessed by buyer $i$, $i \in B$; w.l.o.g. assume that each $m_i >0$.
Let $u_{ij}$ be the utility derived by buyer $i$ on receiving one unit of good $j$. 
Thus, if $x_{ij}$ is the amount of good $j$ that buyer $i$ gets, for $1 \leq j \leq g$, then the total 
utility derived by $i$ is
\[ v_i(x) = \sum_{j=1}^g  {u_{ij} x_{ij} } .\]

The problem is to find prices $\pp = \{p_1, p_2, \ldots , p_g \}$ for the goods so that when each
buyer is given her utility maximizing bundle of goods, the market clears, i.e., each
good having a positive price is exactly sold, without there being any deficiency or surplus. 
Such prices are called {\em market clearing prices} or {\em equilibrium prices}.

The following is the Eisenberg-Gale convex program. Using KKT conditions, one can show that its
optimal solution is an equilibrium allocation for Fisher's linear market and the Lagrange variables
corresponding to the inequalities give equilibrium prices of goods 
(e.g., see Theorem 5.1 in \cite{va.chapter}).

\begin{lp}
\label{CP-EG}
\maximize & \sum_{i \in B} {m_i \log v_i }  \\[\lpskip]
\st       & \forall i\in B:~  v_i = \sum_{j \in G} {u_{ij} x_{ij} }  \nonumber  \\
          & \forall j\in G:~ \sum_{i\in B} x_{ij} \leq 1  \nonumber \\
          & \forall i \in B, \ \forall j\in G:~ x_{ij} \geq 0 \nonumber
\end{lp}


Next, we state the definition of Eisenberg-Gale markets as given in \cite{JV.EG}. 
Let us say that a convex program is an {\em Eisenberg-Gale-type} convex program if its objective function is of the form
\[ \max \sum_{i \in B} {m_i \log v_i } , \]
subject to linear packing constraints, i.e., constraints of the form $\leq$ in which all coefficients and the r.h.s. are
non-negative. Let $\CM$ be a Fisher market, with an arbitrary utility function, whose set of feasible allocations and buyers' utilities 
are captured by a polytope $\Pi$. We will assume that the linear constraints defining $\Pi$ are packing constraints. As a result,
$\CM$ satisfies the {\em free disposal property}, i.e., if $\vv$ is a feasible utility vector then so is any vector dominated
by $\vv$. We will say that an allocation $x_1, \ldots, x_n$ made to the buyers is a {\em clearing allocation} if it uses up
all goods exactly to the extent they are available in $\CM$. Finally, we will say that $\CM$ is an {\em Eisenberg-Gale market}
if any clearing allocation $x_1, \ldots, x_n$ that maximizes
\[ \max \sum_{i \in B} {m_i \log v_i(x_i) }  \]
is an equilibrium allocation, i.e., there are prices $p_1, \ldots p_g$ for the goods such that for each buyer $i$, $x_i$
is a utility maximizing bundle for $i$ at these prices. 
The class EG(2), defined in \cite{CDV.EG}, is essentially the restriction of Eisenberg-Gale markets to the case of 2 buyers; see
\cite{CDV.EG} for the precise definition.

\section{The Feasible Polytope and Some Basic Procedures}
\label{sec.faces}

Let $\CG$ be a game in LNB2 whose solution is captured by convex program (\ref{CP-LNB2}).
We can test if $\CG$ is feasible by solving the following LP:

\begin{lp}
\label{lp.feasible}
\maximize & \ \ \ \ \ \ \ \  \ \ \ \ \ \ \ \ \      t  \\[\lpskip]
\st       & \ \ \ \ \ \ \ \  \ \ \ \ \ \ \ \ \      v_1 \geq  c_1 + t         \nonumber \\
          & \ \ \ \ \ \ \ \  \ \ \ \ \ \ \ \ \      v_2 \geq  c_2 + t            \nonumber \\
\st       & \ \ \ \   \AAA \xx +  \bb_1 v_1  +  \bb_2 v_2   \leq   \ee   \nonumber  \\
          & \ \ \ \ \ \ \ \ \ \ \ \ \ \ \ \ \ \ \ \ \ \ \ \ \ \xx \geq 0   \nonumber 
\end{lp}          
          
Now, $\CG$ is feasible iff the optimal value of $t > 0$. Henceforth, assume that $\CG$ is feasible.

Next, we make the following change of variables,
\[ \mbox{for} \ i = 1, 2: \ \ y_i = v_i - c_i , \] 
hence obtaining the following program which is equivalent to (\ref{CP-LNB2}).

\begin{lp}
\label{CP-y}
\maximize & \sum_{i = 1,2} w_i  \log {y_i}  \\[\lpskip]
\st       & \ \ \ \   \AAA \xx +  \bb_1  (y_1 + c_1)   +  \bb_2 (y_2 + c_2)   \leq   \ee   \nonumber  \\
          & \ \ \mbox{for} \ \ i = 1, 2: \ \ \ \ \  y_i  \geq  0   \nonumber  \\
          & \ \ \ \ \ \ \ \ \ \ \ \ \ \ \ \ \ \ \ \ \ \ \ \ \ \xx \geq 0   \nonumber 
\end{lp}

Henceforth, we will denote $(\ee - c_1 \bb_1 - c_2 \bb_2)$ by $\dd$.
We will denote by $\Pi$ the polyhedron in $\R^{n + 2}$ which is defined by the
constraints of program (\ref{CP-y}). In this paper, we will write the constraints of (\ref{CP-y}) concisely
as follows. This notation will also be used for LP's optimizing over the polytope $\Pi$.

\begin{lp}
\label{CP-concise}
\maximize & \sum_{i = 1,2} w_i  \log {y_i}  \\[\lpskip]
\st       & \ \ \ \ \ \ \ \ \ \    (\xx, y_1, y_2) \in \Pi   \nonumber
\end{lp}

The projection of $\Pi$ onto the coordinates $y_1, y_2$ gives a polytope,
$\CN$ in $\R^2$, which we will call the {\em feasible polytope}; since $\CG$ is feasible, this polytope is full dimensional. 
In this section, we will describe its {\em useful facets}, i.e., 
facets on which the solution to game $\CG$ can lie, and we will give some basic procedures for operating on these facets.

We first compute the point $(l_1, l_2)$ by first maximizing $y_1$ over $\Pi$ to get $l_1$ and then maximizing $y_2$ over 
$\Pi$, subject to $y_1 = l_1$, to get $l_2$. 
Similarly, compute the point $(h_1, h_2)$ by first maximizing $y_2$ over $\Pi$ to get $h_2$ and then maximizing $y_1$ over 
$\Pi$, subject to $y_2 = h_2$, to get $h_1$. Clearly, both these points are vertices of $\CN$. Finally, the set of facets encountered in
moving, on the boundary of $\CN$, from $(l_1, l_2)$ to $(h_1, h_2)$, by increasing the second coordinate are the useful facets.

Each of the useful facets has the form
\[ y_1 + \al y_2  \leq  \bt  , \]
where $\al > 0$ and $\bt > 0$. We will denote the vertex at the intersection of the two facets
\[ y_1 + \al_1 y_2  \leq  \bt_1  \ \ \  \mbox{and} \ \ \  y_1 + \al_2 y_2  \leq  \bt_2  , \]
by $(\al_1, \al_2)$; we will assume $\al_1 < \al_2$.

Let $\al^1$ and $\al^2$ be the $\al$ values of the first and last facets encountered in moving
from $(l_1, l_2)$ to $(h_1, h_2)$; clearly, $\al^1 < \al^2$. Our binary search will be performed on the interval $[\al^1, \al^2]$.
In Procedure 3 below, we show how to compute $\al^1$ and $\al^2$.

The solution to $\CG$ must lie on a face which is either a useful facet or
a vertex at the intersection of 2 useful facets. These 2 possibilities give rise to 
distinct procedures and proofs throughout.

\subsection{Procedure 1: Given $\al$, find the face it lies on}
\label{sec.P1}

We give an algorithm for the following task: 
Given a number $\al$ s.t. $\al^1 \leq \al \leq \al^2$, determine which of the following possibilities holds:
\begin{enumerate}
\item
$\al$ defines a facet of $\CN$, $y_1 + \al y_2 \leq \bt$, for a suitable value of $\bt$.
If so, find this facet.
\item
There is a vertex of $\CN$, $(\al_1, \al_2)$, such that $\al_1 < \al < \al_2$.
If so, find this vertex.
\end{enumerate}

First solve the following LP and let its optimal objective function value be denoted by $\bt$ and
let $a$ and $b$ denote the optimal values of $y_1$ and $y_2$, respectively.

\begin{lp}
\label{alpha}
\maximize & \ \ \ \ \ \ \ \  \ \  y_1 + \al y_2     \\[\lpskip]
\st       & \ \ \ \ \ \ \ \ \ \    (\xx, y_1, y_2) \in \Pi   \nonumber
\end{lp}

Having computed $\bt$, solve the following LP and let its objective function value be denoted 
by $a_1$.

\begin{lp}
\label{min}
\minimize & \ \ \ \ \ \ \ \  \ \ \ \   y_1      \\[\lpskip]
\st       &  \ \ \ \ \ \ \ \ \    y_1 + \al y_2 = \bt                \nonumber \\
          & \ \ \ \ \ \ \ \ \ \    (\xx, y_1, y_2) \in \Pi   \nonumber
\end{lp}

Next, change the objective in LP (\ref{min}) to maximize $y_1$, and let its optimal objective function value be $a_2$.
If $a_1 < a_2$, we are in the first case. Define $b_1 = (\bt - a_1)/\al$ and $b_2 = (\bt - a_2)/\al$.
Then, the endpoints of the facet $y_1 + \al y_2 = \bt$ are $(a_1, b_1)$ and $(a_2, b_2)$.
Otherwise, $a_1 = a_2 = a$, say, and we are in the second case. Let $b$ be the value of $y_2$ computed in LP (\ref{min}).
Then, the vertex has coordinates $(a, b)$. 

Next, we need to find $\al_1$ and $\al_2$ for this vertex. Let us begin by writing the dual for LP (\ref{alpha}).

\begin{lp}
\label{dual}
\minimize & \ \ \ \ \ \ \ \     \sum_j {d_j p_j}      \\[\lpskip]
\st       &  \ \ \ \     \ \ \  \sum_j {b_{1j} p_j}   \geq 1   \nonumber  \\
          &  \ \ \ \     \ \ \  \sum_j {b_{2j} p_j}   \geq \al    \nonumber \\
          & \mbox{for} \ \ 1 \leq i \leq n:  \  \  \sum_j {A_{ji} p_j}  \geq  0  \nonumber  \\
          & \mbox{for} \ \ 1 \leq j \leq m:  \  \   p_j  \geq  0  \nonumber  
\end{lp}

Let $(\xx^*, y_1^*, y_2^*)$ be an optimal solution to LP (\ref{alpha}). Since $\CG$ has been assumed to be feasible, 
$y_1^* > 0$ and $y_2^* >0$. 
The next LP is derived from LP (\ref{dual}) by adding constraints on $p_j$ which are implied by the
complementary slackness conditions of the primal and dual pair of LP's (\ref{alpha}) and (\ref{dual}).
It is not optimizing any function, since we are only concerned with its feasible solutions.

\begin{lp}
\label{LP-alpha}
          &  \ \ \ \     \ \ \  \sum_j {b_{1j} p_j}   =   1   \\[\lpskip]
          &  \ \ \ \     \ \ \  \sum_j {b_{2j} p_j}   =   r    \nonumber \\
          & \mbox{for} \ \ 1 \leq i \leq n:  \  \  \sum_j {A_{ji} p_j}  \geq  0  \nonumber  \\
          & \mbox{for} \ \ 1 \leq i \leq n \ s.t. \ x_i^* > 0 :  \  \  \sum_j {A_{ji} p_j}  =  0  \nonumber  \\
          & \mbox{for} \ \ 1 \leq j \leq m \ s.t. \ \sum_i {A_{ji} x_i^*} + b_{ij} y_1^* + b_{2j} y_2^* < d_j :  \ \ \  \   p_j  =  0  \nonumber  \\
          & \mbox{for} \ \ 1 \leq j \leq m:  \  \   p_j  \geq  0  \nonumber  
\end{lp}

The next lemma follows from the 
complementary slackness conditions of the primal and dual pair of LP's (\ref{alpha}) and (\ref{dual}).

\begin{lemma}
\label{lem.alpha-range}
$ \{ \al \ | \ \mbox{LP} \ (\ref{alpha}) \  \mbox{attains its optimal solution at} \  (a, b) \} $ \\
$ = \  \{ r \ | \ \exists \ \mbox{a feasible solution to LP} \ (\ref{LP-alpha}) \  \mbox{in which} \ \sum_j {b_{2j} p_j}   =   r   \} .$
\end{lemma}
 
\begin{proof}
Let $A$ and $B$ denote the sets on the l.h.s. and r.h.s. of the equality, respectively.
Let $f \in A$. Clearly, $(\xx^*, a, b)$ is an optimal solution to LP (\ref{alpha}) with $\al$ substituted by $f$.
Let $\pp^*$ be an optimal solution to LP (\ref{dual}). Since $\pp^*$ satisfies all complementary slackness conditions of 
LP's (\ref{alpha}) and (\ref{dual}), it is a feasible solution to LP (\ref{LP-alpha}) with $r$ replaced by $f$. Hence $f \in B$
and $A \subseteq B$.

The reverse inclusion follows in a similar manner, again using complementary slackness conditions of 
LP's (\ref{alpha}) and (\ref{dual}).
\end{proof}
 
By Lemma \ref{lem.alpha-range}, we can obtain $\al_1$ and $\al_2$ as follows. First, minimize $r$ subject to the
constraints of LP (\ref{LP-alpha}); this gives $\al_1$. Next, maximize $r$ subject to the
constraints of LP (\ref{LP-alpha}); this gives $\al_2$.

\subsection{Procedure 2: Given $(a, b)$, find the face it lies on}
\label{sec.P2}

Given a point $(a, b)$ on the boundary of $\CN$, we give a procedure for finding the facet or
vertex it lies on. First, solve LP (\ref{LP-flows}) for finding a feasible allocation gives the 2 buyers utilities of
$y_1 = a$ and $y_2 = b$.

\begin{lp}
\label{LP-flows}
          & \ \ \ \ \ \ \ \  \ \ \ \   y_1  = a     \\[\lpskip]
          & \ \ \ \ \ \ \ \  \ \ \ \   y_2  = b     \nonumber  \\
          & \ \ \ \ \ \ \ \ \ \    (\xx, y_1, y_2) \in \Pi   \nonumber
\end{lp}

Next, solve the minimization and maximization
versions, with objective function $r$, of LP (\ref{LP-alpha}) to find $\al_1$ and $\al_2$, respectively. 
If $\al_1 = \al_2 = \al$, $(a, b)$ lies on the
facet $y_1 + \al y_2 \leq a + \al b$. Otherwise, $\al_1 < \al_2$ and $(a, b)$ lies on the vertex $(\al_1, \al_2)$.

\subsection{Procedure 3: Computing $\al^1$ and $\al^2$}
\label{sec.P3}

We now show how to compute $\al^1$ and $\al^2$, defined at the beginning of this section. As stated there,
our binary search will be performed on the interval $[\al^1, \al^2]$.

First, use Procedure 2 to find the vertex, say $(\al_1, \al_2)$, on which $(l_1, l_2)$ lies. Set, $\al^1 \la \al_1$.
Next, use Procedure 2 to find the vertex, say $(\al_1, \al_2)$, on which $(h_1, h_2)$ lies. Set, $\al^2 \la  \al_2$.

\section{Binary Search on Parameter $z$}
\label{sec.bs}

We first give some crucial definitions. Let $(f_1, f_2)$ be the solution to game $\CG$. 
For player $i$ define 
\[ \ga_i = {f_i \over m_i} . \]
Define parameter $z$ to be 
\[ z = {\ga_1 \over \ga_2} .\]
The next lemma relates $z$ to the point where the solution lies.

\begin{lemma}
\label{lem.z}
If the solution to game $\CG$ lies on:
\begin{enumerate}
\item
the facet $y_1 + \al y_2  \leq  \bt$, then $z = \al$.
\item
the vertex $(\al_1, \al_2)$, then $\al_1 < z < \al_2$.
\end{enumerate}
\end{lemma}

\begin{proof}
In the first case, the objective function of the convex program (\ref{CP-2D}),
\[ g = w_1 \log{y_1} + w_2 \log {y_2} \]
must be tangent to the facet at the solution point, say $(a, b)$.
Equating  the ratio of the partial derivatives of $g$ and the line $y_1 + \al y_2 = \bt$ w.r.t. $y_2$ and $y_1$, we get
\[ {{a/w_1} \over {b/w_2}} = \al .\]
But the l.h.s. is $\ga_1/\ga_2 = z$, thereby giving $z = \al$.

In the second case, the tangent to $g$ at the solution must be intermediate
between the slopes of the adjacent facets, giving $\al_1 < z < \al_2$.
\end{proof}

Our algorithm will conduct a binary search on $z$, on the interval $[\al^1, \al^2]$,
to find the right face on which the solution lies. The test given in the next lemma helps determine, in each iteration,  
if the current face is the right.

\begin{lemma}
\label{lem.m}

\begin{enumerate}
\item
The solution to game $\CG$ lies on
the facet $y_1 + \al y_2  \leq  \bt$, having endpoints $(a_1, b_1)$ and $(a_2, b_2)$, with $a_1 < a_2$,  iff
\[ {w_1 \over {w_1 + w_2}} \in \left[{a_1 \over \bt} , {a_2 \over \bt} \right] .\]
\item
The solution to market $\CM$ lies on
the vertex $(\al_1, \al_2)$, having coordinates $(a, b)$, which is at the intersection of 
facets $y_1 + \al_1 y_2  \leq  \bt_1$ and $y_1 + \al_2 y_2  \leq  \bt_2$, iff
\[ {w_1 \over {w_1 + w_2}} \in \left({a \over \bt_2} , {a \over \bt_1} \right) .\]
\end{enumerate}
\end{lemma}

\begin{proof}
In the first case, substituting $w_i = f_i/\ga_i$ and $\al = \ga_1/\ga_2$ (this follows from Lemma \ref{lem.z}), we get
\[ {w_1 \over {w_1 + w_2}} =  {f_1 \over {f_1 + \al f_2}}  =  {f_1 \over \bt} \in \left[{a_1 \over \bt} , {a_2 \over \bt} \right] . \]
For the other direction, if $w_1 / (w_1 + w_2)$ lies in the interval given,
then by the equation given above, $f_1 \in [a_1, a_2]$, thereby showing that the solution lies on the facet $y_1 + \al y_2  \leq  \bt$.

In the second case, by Lemma \ref{lem.z}, $\ga_1/\ga_2 \in (\al_1, \al_2)$, and this leads to the interval in which
$w_1/(w_1 + w_2)$ lies. The proof of the other direction also follows in the same manner.
\end{proof}

The operation in Step \ref{step.repeat} in Algorithm
\ref{alg.main}, $\lfloor x \rfloor_{\kappa}$, 
truncates $x$ to accuracy $2^{-\kappa}$ where $\kappa$ is defined in the proof of Lemma \ref{lem.iterations}.

\begin{lemma}
\label{lem.iterations}
Binary search executes polynomial in $n$ iterations.
\end{lemma}

\begin{proof}
First, we place an upper bound on the size of the interval $[\al^1, \al^2]$. By Cramer's rule, the number of bits in the solution to
LP (\ref{LP-alpha}) is polynomial in $n$. Let this number be $\kappa$. Therefore, for each of the facets, $\al$ can be written in
$\kappa$ bits. However, we do not know where the binary point lies. So, let us assume that we will only deal with $2 \kappa$ bit
long numbers, with $\kappa$ bits before and $\kappa$ bits after the binary point. The operation in Step \ref{step.repeat} in Algorithm
\ref{alg.main}, $\lfloor x \rfloor_{\kappa}$, 
is meant to truncate $x$ to this form. Therefore, the size of the interval is bounded by $2^{2 \kappa}$. 
Hence binary search will execute $O(\kappa)$, i.e., polynomial in $n$ iterations.
\end{proof}

\bigskip

\noindent

\fbox{
\begin{algorithm}{\label{alg.main} (Binary Search)}

\step
\label{step.phase2}
{\bf (Initialization:)} 
$l \ \lefta \ \al^1$ \ \ and  \ \ $h \ \lefta \ \al^2$. \\
\[ \mbox{Let} \ r \la {w_1 \over {w_1 + w_2}}  . \]

\step
\label{step.repeat}
\[ \al \ \lefta \ \lfloor {{l + h} \over 2} \rfloor_{\kappa} . \]

\step
Using Procedure 1 (Section \ref{sec.P1}), determine if $\al$ lies on:

\begin{description}
\item
{\bf Case 1:} \  \  A facet, say $y_1 + \al y_2 \leq \bt$, with
endpoints $(a_1, b_1)$ and $(a_2. b_2)$. \\

If $r < (a_2 / \bt)$ then 
$l  \ \la \ \al$ and go to step \ref{step.repeat}. \\

Else if $r > (a_1 / \bt)$ then 
$h \ \la \ \al$  and go to step \ref{step.repeat}. \\

\[ \mbox{Else if} \ r \in \left[{a_1 \over \bt} , {a_2 \over \bt} \right], \
\mbox{then solve the following 2 equations for} \ y_1 \ \mbox{and} \ y_2: \] 
\[ y_1 + \al y_2 = \bt  \ \ \ \ \mbox{and} \ \ \ \ 
 {{y_1 /w_1} \over {y_2 /w_2}}  =  \al .\]
If the solution is $y_1 = a, y_2 = b$, output the solution to game $\CG$:  \\
$v_1 = a + c_1$ and $v_2 = b + c_2$, and HALT.

\item
{\bf Case 2:} \  \  A vertex, say $(\al_1, \al_2)$, with coordinates $(a, b)$, \\

If $r \leq \ (a / \bt_2)$ then 
$l \ \lefta \ \al_2$ and go to step \ref{step.repeat}. \\

Else if $r \geq  (a / \bt_1)$ then 
$h \ \lefta \ \al_1$  and go to step \ref{step.repeat}. \\

\[ \mbox{Else if} \ r \in \left({a \over \bt_2} , {a \over \bt_1} \right) ,\]
then output the solution to game $\CG$: $v_1 = a + c_1$ and $v_2 = b + c_2$, and HALT.

\end{description}

\step
End.

\end{algorithm}
}

\bigskip

\begin{lemma}
\label{lem.bs}
Algorithm \ref{alg.main} performs binary search correctly.
\end{lemma}

\begin{proof}
We first justify restricting search to $\al$-values in the range $[\al^1, \al^2]$, i.e., the faces encountered on the boundary of
$\CN$ in moving from $(l_1, l_2)$ to $(h_1, h_2)$ by increasing the socond coordinate, as stated in Section \ref{sec.faces}.
Observe that starting at $(l_1, l_2)$ and moving in the other direction on the boundary of $\CN$, the second coordinate must
decrease and the first either remains the same or decreases, hence decreasing the objective function value of (\ref{CP-LNB2}).
Similarly, moving beyond $(h_1, h_2)$ on the boundary of $\CN$, the first coordinate must
decrease and the second either remains the same or decreases, again decreasing the objective function value of (\ref{CP-LNB2}).

By the necessary and sufficient conditions established in Lemma \ref{lem.m}, the algorithm can determine if the current
face is the right one or not. Next, observe that 
value of $w_1 /  (w_1 + w_2)$ decreases monotonically in moving from $(l_1, l_2)$ to $(h_1, h_2)$ on the boundary of $\CN$.
With this observation, one can check that if the current face is not the right one, in each case, the algorithm moves to the side of 
this face that contains the right face.
\end{proof}

Hence we get:

\begin{theorem}
\label{thm.LNB2}
Every game in LNB2 has a rational solution; moreover, such a solution can be found in polynomial time
using only an LP solver.
\end{theorem}

Next assume that the coefficients in the constraints of convex program (\ref{CP-LNB2}) are ``small'', i.e., polynomially bounded in $n$.
Then all LP's that need to be solved will also have ``small'' coefficients
(the objective function and right hand side don't need to be ``small''). Such LP's can be solved in strongly polynomial time \cite{tardos.LP}.
By Lemma \ref{lem.iterations}, binary search will execute only polynomial in $n$ iterations. 
Hence we get:

\begin{theorem}
\label{thm.SLNB2}
Every game in SLNB2 can be solved in strongly polynomial time.
\end{theorem}

In particular, the game \G2, which lies in SLNB2, can be solved in strongly polynomial time.
\cite{JV.EG} give examples of Eisenberg-Gale markets with 3 buyers which do not have rational solutions.
Hence, instances of the corresponding Nash bargaining games, with zero disagreement utilities, do not possess rational solutions;
this includes the modification of \G2 \ to a 3-player game, with 3 source-sink pairs.

\section{A Strongly Polynomial Algorithm for the Game ADNB2}
\label{sec.ADNB2}

We first give the notion of a {\em flexible budget market}, introduced 
in \cite{va.NB}. This is a natural variant on a Fisher market -- the main difference is that instead of having a fixed amount
of money to spend in the market, buyers have a (strict) lower bound on the amount of utility they wish to derive, and at any given prices,
they want to derive it in the most cost-effective manner. Thus, the money spent by buyers is a function of prices of goods.
The object again is to find market clearing or equilibrium prices.

Next, we give a precise definition of the flexible budget market, $\CM$, that \ADNB2 \ reduces to (the reduction, derived by applying
KKT conditions, appears in Theorem 4 in \cite{va.NB}).
The goods and utility functions of the two buyers are as in \ADNB2 \ and each buyer $i$ has a parameter $c_i$ giving
a strict lower bound on the amount of utility she wants to derive. 
Given prices $\pp$ for the goods, define the {\em maximum bang-per-buck} of buyer $i$ to be
\[ \ga_i = max_j \left\{ {u_{ij} \over p_j} \right\} .  \]
Now, the amount of money buyer $i$ spends is defined to be 
\[m_i = 1 + {c_i \over \ga_i} . \]

\subsection{The algorithm}

We will first renumber the goods.
Compute $u_{1j} / u_{2j}$ for each good $j$, sort the goods in decreasing order of this ratio and partition by 
equality. For the purpose of this algorithm, it will suffice to replace each partition by one good. Consider a partition and
compute $min_j \{u_{1j} \}$ for goods $j$ in this partition. Assume the minimum is attained by $u_{1k}$. Then the utilities of the
two players for this new good, say $g'$ will be $u_{1k}$ and $u_{2k}$, respectively. Next, we need to compute the number of units
of $g'$ that are available. Each good $j$ in the partition will be represented by $u_{1j}/u_{1k}$ units of $g'$. The sum over all goods in
the partition is the total number of units of this good. 
Let us assume that after this transformation, we have $n$ goods available, $1, 2, \ldots, n$ and the amount of good $j$ is $b_{j}$
and the goods are numbered in decreasing order of $u_{1j}/u_{2j}$.

Next, we test for feasibility, i.e., we need to determine whether the two players can be given baskets providing $c_1$
and $c_2$ utility, respectively, without exhausting all goods. Clearly, the most efficient way of doing this is to give player 1
goods from the lowest index and to give player 2 goods from the highest index. Assume that player 1 needs to be given all the
available goods $1, 2, \ldots k_1 -1$ and an amount $x$ of good $k_1$ in order to make up $c_1$ utility. Next, assume that player 2
needs to be given all available goods $n, n-1, \ldots, k_2 +1$ and an amount $y$ of good $k_2$ to make up $c_2$ utility.
Then, the game and the market are feasible iff $k_1 < k_2$ or $k_1 = k_2$ and $x + y < b_{k_1}$.

Finally, assume that the given market is feasible and let us find an equilibrium for it. Since each buyer must get a utility maximizing
bundle of goods, for each good $j$ that is allocated to player $i$, 
\[ \ga_i = {u_{ij} \over p_j} \]
and for each good $j$ that is not allocated to player $i$,
\[ \ga_i \geq {u_{ij} \over p_j} . \]
This leads to two cases for the equilibrium allocation:
\begin{itemize}
\item
{\bf Case 1:}
There is a $k$, $1 \leq k \leq n$ such that player 1 gets goods $1, 2, \ldots, k$ and player 2 gets goods
$k+1, k+2, \ldots, n$.
\item
{\bf Case 2:}
There is a $k$, $1 \leq k \leq n$ such that player 1 gets goods $1, 2, \ldots, k-1$, player 2 gets goods
$k+1, k+2, \ldots, n$, and they both share good $k$.
\end{itemize}
Since the equilibrium prices are unique, only one of these $O(n)$ possibilities holds. We will check them all to
determine which one it is.

{\bf Case 1:}
Let $G_1$ consist of the first $k$ good and $G_2$ consist of the rest. Then,
\[ \ga_1 = {u_{1j} \over p_j} \ \ \mbox{for} \ \ j \in G_1 \ \ \mbox{and} \ \ \ga_2 = {u_{2j} \over p_j} \ \ \mbox{for} \ \ j \in G_2 . \]
Let $\ga_1 = 1/x$ and $\ga_2 = 1/y$. The total money spent by player 1 is
\[ m_1 = \sum_{j \in G_1} {p_j b_j} = x \sum_{j \in G_1} {u_{1j} b_j} = w_1 + c_1 x .\]
Similarly, the total money spent by player 2 is
\[ m_2 = \sum_{j \in G_2} {p_j b_j} = y \sum_{j \in G_2} {u_{2j} b_j} = w_2 + c_2 y .\]
Solve these equations for $x$ and $y$ and compute the prices of goods $p_j$. If with these prices, each player gets a utility maximizing 
bundle of goods, i.e., the 2 conditions given above hold, these are equilibrium prices and allocations.

{\bf Case 2:}
Since good $k$ is allocated to both buyers, 
\[ \ga_1 = {u_{1k} \over p_k}  \ \ \mbox{and} \ \ \ga_2 = {u_{2k} \over p_k}  . \]
Let $u_{1k}/u_{2k} = \al$ and $\ga_1 = 1/x$. Then $\ga_2 = 1/(\al x)$.
Let $G_1$ consist of the first $k$ good and $G_2$ consist of the rest. Then the total money spent by both players is
\[ m_1 + m_2 = \sum_{j \in G} {p_j b_j} = x (\sum_{j \in G_1} {u_{1j} b_j} + \sum_{j \in G_2} {\al u_{2j} b_j}) = w_1 + c_1 x + w_2 + c_2 \al x .\]
Again, solve for $x$, compute prices of goods and check if the conditions for equilibrium are satisfied.

Observe that \ADNB2 \ is not in SLNB2, since the $u_{ij}$'s are not restriced to be polynomially bounded in $n$. Even so,
we get:

\begin{theorem}
\label{thm.ADNB2}
There is a combinatorial strongly polynomial algorithm for solving \ADNB2.
\end{theorem}

\section{The Circle Game}
\label{sec.circle-alg}

Using the KKT conditions of (\ref{CP-circle}) it is easy to show that the Nash bargaining solution $(x, y)$ satisfies 
the following equations:
\[ (2 y^2 - c_2y -1)^2  = c_1^2 (1 - y^2)  \ \ \ \  \mbox{and} \ \ \ \  x^2 + y^2 = 1 .\]

On the other hand, the problem also has a simple geometric solution. 
Let Q be the point on the unit circle in the positive orthant.
Let O denote the origin and P denote the point $(c_1, c_2)$.
Let $\theta_1$ be the angle made by PQ with the $x$-axis and $\theta_2$ be the angle made by OQ with the $y$-axis.

\begin{proposition}
\label{prop.circle}
Q is the Nash bargaining solution iff $\theta_1 = \theta_2$.
\end{proposition}

\begin{proof}
Let $(a, b)$ be the point Q and let R be the intersection of the vertical line passing through Q and the horizontal line passing
through P. Then the angle QPR is $\theta_1$.

The slope of the tangent to the hyperbola $(x - c_1) (y - c_2) = \al$ at $(x, y)$, which is obtained by
taking ratio of partial derivatives w.r.t. $y$ and $x$, is
\[ {{y - c_2} \over {x - c_1}} .\]
From the triangle PQR we get that
\[ \tan {\theta_1} = {{b - c_2} \over {a - c_1}} .\]
The slope of the tangent to the circle at Q is $\tan {\theta_2}$.

By Nash's theorem, Q is the Nash bargaining solution iff
the hyperbola $(x - c_1) (y - c_2) = \al$ is tangent to the unit circle at point Q, for a suitable value of $\al$.
Hence, by the above-stated facts, Q is the Nash bargaining solution iff $\theta_1 = \theta_2$.
\end{proof}

\section{Acknowledgment}
A special thanks to Inbal Talgam for pointing out that my previous approach to Theorem \ref{thm.LNB2} was needlessly complicated and 
to Leonard Schulman for the pointing out the attractive geometric solution to the circle game.

\end{document}